\documentclass[11pt]{article}
\usepackage{graphicx,amsmath,amsfonts,amssymb,bm,hyperref,url,breakurl,epsfig,epsf,color,fullpage,MnSymbol,mathbbol, fmtcount,algorithmic,algorithm, semtrans
} 

\usepackage{titlesec}

\usepackage{tikz}
\usepackage{pgfplots}

\usetikzlibrary{pgfplots.groupplots}

\titleformat{\paragraph}
{\normalfont\normalsize\bfseries}{\theparagraph}{1em}{}
\titlespacing*{\paragraph}
{0pt}{3.25ex plus 1ex minus .2ex}{1.5ex plus .2ex}

\usepackage{caption,subcaption}
\usepackage[bottom,hang,flushmargin]{footmisc} 

\usepackage{hyperref}
\definecolor{darkred}{RGB}{150,0,0}
\definecolor{darkgreen}{RGB}{0,150,0}
\definecolor{darkblue}{RGB}{0,0,200}
\hypersetup{colorlinks=true, linkcolor=darkred, citecolor=darkgreen, urlcolor=black}

\newtheorem{theorem}{Theorem}[section]
\newtheorem{lemma}[theorem]{Lemma}

\newcommand\tr{{{\operatorname{tr}}}}

\newcommand{\fronorm}[1]{\left\|#1\right\|_{F}}
\newcommand{\onenorm}[1]{\left\|#1\right\|_{\ell_1}}
\newcommand{\twonorm}[1]{\left\|#1\right\|_{\ell_2}}

\newcommand{\abs}[1]{\left|#1\right|}

\newcommand{\cA}{\mathcal{A}}
\newcommand{\x}{\vct{x}}
\newcommand{\y}{\vct{y}}
\newcommand{\W}{\mtx{W}}

\newcommand{\R}{\mathbb{R}}
\newcommand{\C}{\mathbb{C}}

\newcommand{\<}{\langle}
\renewcommand{\>}{\rangle}

\renewcommand{\P}{\operatorname{\mathbb{P}}}
\newcommand{\E}{\operatorname{\mathbb{E}}}

\newcommand{\vct}[1]{\bm{#1}}
\newcommand{\mtx}[1]{\bm{#1}}

\newcommand{\rank}{\operatorname{rank}}

\definecolor{ejc}{RGB}{0,0,255}

\numberwithin{equation}{section} 

\def \endprf{\hfill {\vrule height6pt width6pt depth0pt}\medskip}
\newenvironment{proof}{\noindent {\bf Proof} }{\endprf\par}

\newcommand{\avg}[1]{\left< #1 \right>}

\title{Phase Retrieval from Coded Diffraction Patterns}

\author{ Emmanuel J. Cand\`{e}s\thanks{Departments of Mathematics and
    of Statistics, Stanford University, Stanford CA} \quad  Xiaodong
  Li\thanks{Department of Statistics, The Wharton School, University
    of Pennsylvania, Philadelphia, PA} \quad Mahdi
  Soltanolkotabi\thanks{Department of Electrical Engineering, Stanford
    University, Stanford CA} }

\begin{document}
\maketitle

\begin{abstract}
  This paper considers the question of recovering the phase of an
  object from intensity-only measurements, a problem which naturally
  appears in X-ray crystallography and related disciplines. We study a
  physically realistic setup where one can modulate the signal of
  interest and then collect the intensity of its diffraction pattern,
  each modulation thereby producing a sort of {coded diffraction
    pattern}. We show that {\em PhaseLift}, a recent convex
  programming technique, recovers the phase information exactly from a
  number of random modulations, which is polylogarithmic in the number
  of unknowns.  Numerical experiments with noiseless and noisy data
  complement our theoretical analysis and illustrate our approach.
\end{abstract}

\section{Introduction}

\subsection{The phase retrieval problem}

In many areas of science and engineering, we only have access to
magnitude measurements; for instance, it is far easier for detectors
to record the modulus of the scattered radiation than to measure its
phase. Imagine then that we have a discrete object $\x \in \C^n$, and
that we would like to measure $\<\vct{a}_k, \x\>$ for some sampling vectors
$\vct{a}_k \in \C^n$ but only have access to phaseless measurements of the
form
\begin{equation}
\label{eq:general}
y_k = |\<\vct{a}_k, \x\>|^2, \quad k = 1, \ldots, m. 
\end{equation}
The phase retrieval problem is that of recovering the missing phase of
the data $\<\vct{a}_k, \x\>$. Once this information is available, one
can find the vector $\x$ by essentially solving a system of linear
equations.

The quintessential phase retrieval problem, or phase problem for
short, asks to recover a signal from the modulus of its Fourier
transform. This comes from the fact that in coherent X-ray imaging, it
follows from the Fraunhofer diffraction equation that the optical
field at the detector is well approximated by the Fourier transform of
the object of interest. Since photographic plates, CCDs and other
light detectors can only measure light intensity, the problem is then
to recover $\x = \{x[t]\}_{t = 0}^{n-1} \in \C^n$ from measurements of
the type
\begin{equation}
  \label{eq:Fourier}
  y_k = \left| \sum_{t = 0}^{n-1} x[t] e^{-i2\pi \omega_k t}
  \right|^2, \quad \omega_k \in \Omega, 
\end{equation}
where $\Omega$ is a sampled set of frequencies in $[0,1]$ (we stated
the problem in one dimension to simplify matters). We thus recognize
an instance of \eqref{eq:general} in which the vectors $\vct{a}_k$ are
sampled values of complex sinusoids.  X-ray diffraction images are of
this form, and as is well known, permitted the discovery of the double
helix \cite{watson1953structure}. In addition to X-ray crystallography
\cite{harrison1993phase, millane1990phase}, the phase problem has
numerous other applications in the imaging sciences such as
diffraction and array imaging \cite{bunk2007diffractive,
  chai2011array}, optics \cite{walther1963question}, speckle imaging
in astronomy \cite{fienup1987phase}, and microscopy
\cite{miao2008extending}. Other areas where related problems appear
include acoustics \cite{balan2006signal, balan2010signal}, blind
channel estimation in wireless communications \cite{ahmed2012blind,
  ranieri2013phase}, interferometry \cite{demanet2013convex}, quantum
mechanics \cite{corbett2006pauli, reichenbach1965philosophic} and
quantum information \cite{heinosaari2013quantum}.

\subsection{Convex relaxation}

Previous work \cite{candes2013phase,chai2011array} suggested to bring
convex programming techniques to bear on the phase retrieval
problem. Returning to the general formulation \eqref{eq:general}, the
phase problem asks to recover $\x \in \C^n$ subject to data
constraints of the form
\[
\tr(\vct{a}_k\vct{a}_k^* \x\x^*) = y_k, \quad k = 1, \ldots, m, 
\]
where $\tr(\mtx{X})$ is the trace of the matrix $\mtx{X}$.  The idea is then to
lift the problem in higher dimensions: introducing the Hermitian
matrix variable $\mtx{X} \in \mathcal{S}^{n \times n}$, the phase problem is
equivalent to finding $\mtx{X}$ obeying
\begin{equation}
  \label{eq:NP}
  \mtx{X} \succeq 0, \quad \rank(\mtx{X}) = 1, \quad  
  \tr(\vct{a}_k\vct{a}_k^* \mtx{X}) = y_k \text{ for }  k = 1, \ldots, m 
\end{equation}
where, here and below, $\mtx{X} \succeq 0$ means that $\mtx{X}$ is
positive semidefinite.  This problem is not tractable and, by dropping
the rank constraint, is relaxed into
\begin{equation}
  \label{eq:PL}
  \begin{array}{ll}
    \text{minimize}   & \quad \tr(\mtx{X})\\
    \text{subject to} & \quad  \mtx{X} \succeq 0\\
    & \quad  \tr(\vct{a}_k\vct{a}_k^* \mtx{X}) = y_k, \quad k = 1, \ldots, m. 
\end{array}
\end{equation}
PhaseLift \eqref{eq:PL} is a semidefinite program (SDP). If its
solution happens to have rank one and is equal to $\vct{x}\vct{x}^*$,
then a simple factorization recovers $\vct{x}$ up to a global
phase/sign.

We pause to emphasize that in different contexts, similar convex
relaxations for optimizing quadratic objectives subject to quadratic
constraints are known as Schor's semidefinite relaxations, see
\cite[Section 4.3]{nemirovski2001lectures} and
\cite{goemans1995improved} on the MAXCUT problem from graph theory for
spectacular applications of these ideas. For related convex
relaxations of quadratic problems, we refer the interested reader to
the wonderful tutorial \cite{luo2010semidefinite}.

\subsection{This paper}

Numerical experiments \cite{candes2013phase} together with emerging
theory suggest that the PhaseLift approach is surprisingly
effective. On the theoretical side, starting with
\cite{candes2012phaselift}, a line of work established that if the
sampling vectors $\vct{a}_k$ are sufficiently randomized, then the
convex relaxation is provably exact. Assuming that the $\vct{a}_k$'s
are independent random (complex-valued) Gaussian vectors,
\cite{candes2012phaselift} shows that on the order of $n \log n$
quadratic measurements are sufficient to guarantee perfect recovery
via \eqref{eq:PL} with high probability. A subset of the authors
\cite{candes2012solving} reached the same conclusion from on the order
of $n$ equations only, by solving the SDP feasibility problem; to be
sure, \cite{candes2012solving} establishes that the set of matrices
obeying the constraints in \eqref{eq:PL} reduces to a unique point
namely, $\x\x^*$, see \cite{demanet2012stable} for a similar
result.\footnote{\cite{candes2012solving} also establishes
  near-optimal estimation bounds from noisy data.} Finally, inspired
by PhaseLift and the famous MAXCUT relaxation of Goemans and
Williamson, \cite{waldspurger2012phase} proposed another semidefinite
relaxation called PhaseCut whose performance from noiseless data---in
terms of the number of samples needed to achieve perfect
recovery---turns out to be identical to that of PhaseLift.

While this is all reassuring, the problem is that the Gaussian model,
in which each measurement gives us the magnitude of the dot product
$\sum_{t = 0}^{n-1} x[t] a_k[t]$ between the signal and
(complex-valued) Gaussian white noise, is very far from the kind of
data one can collect in an $X$-ray imaging and many related
experiments. The purpose of this paper is to show that the PhaseLift
relaxation is still exact in a physically inspired setup where one can
modulate the signal of interest and then let diffraction occur.

\subsection{Coded diffraction patterns}

Imagine then that we modulate the signal before diffraction. Letting
$d[t]$ be the modulating waveform, we would observe the diffraction
pattern
\begin{equation}
  \label{eq:modulated}
  y_k = \left| \sum_{t = 0}^{n-1} x[t] \bar{d}[t] e^{-i2\pi \omega_k t}
  \right|^2, \quad \omega_k \in \Omega.
\end{equation}
We call this a {\em coded diffraction pattern} (CDP) since it gives us
information about the spectrum of $\{x[t]\}$ modulated by the code
$\{d[t]\}$. There are several ways of achieving modulations of this
type: one can use a phase mask just after the sample, see Figure
\ref{fig:xray}, or use an optical grating to modulate the illumination
beam as mentioned in \cite{loewen1997diffraction}, or even use
techniques from ptychography which scan an illumination patch on an
extended specimen \cite{rodenburg2008ptychography,
  thibault2009probe}. We refer to \cite{candes2013phase} for a more
thorough discussion of such approaches.

\begin{figure}
\centering
\begin{tikzpicture}
\node at (0,0) {\includegraphics[width=0.7\linewidth]{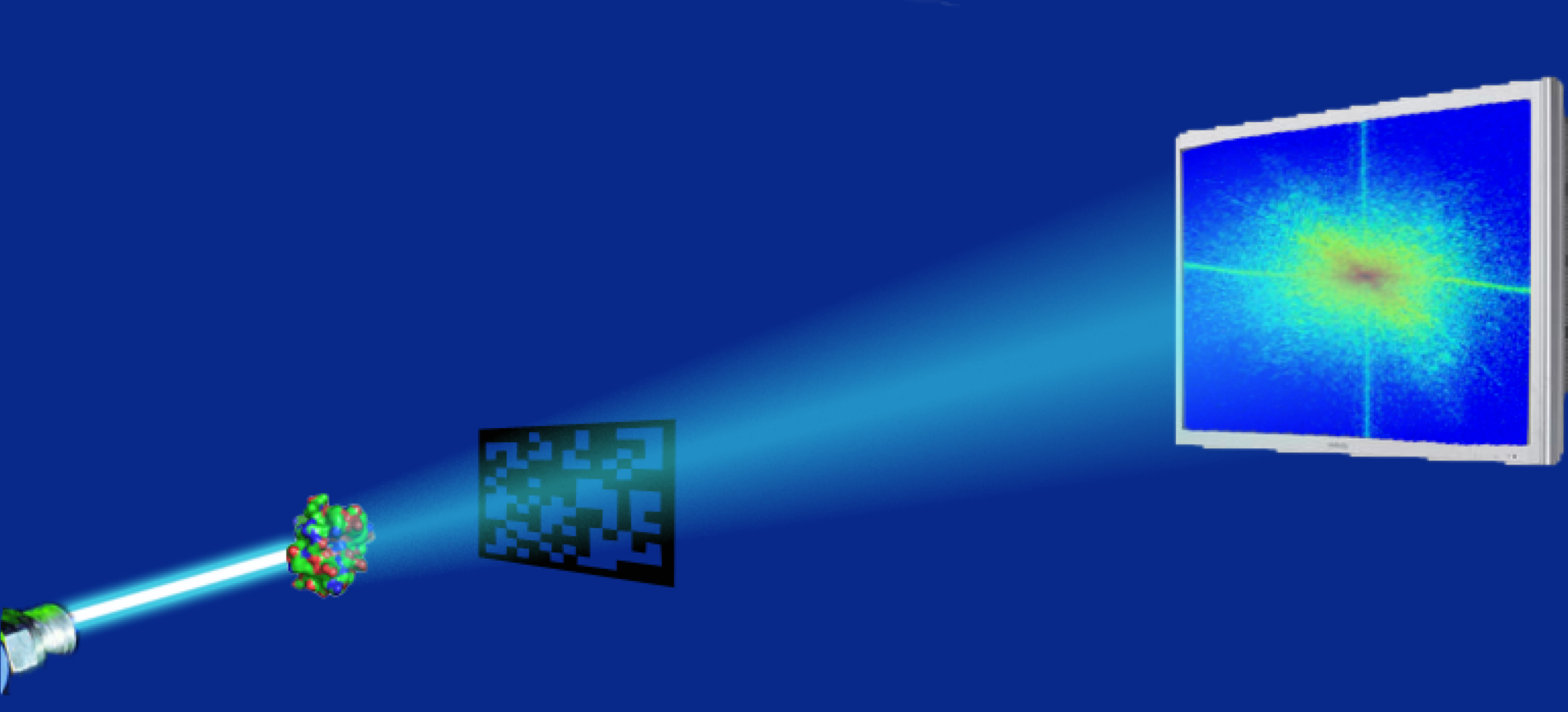}};
\node[white] at (-5.2,-1.3) {source};
\node[white] at (-3.3,-1.8) {sample};
\node[white] at (-1.5,-1.8) {phase plate};
\node[white] at (4,-0.8) {diffraction patterns};
\end{tikzpicture}
\label{fig:xray}
\caption{An illustrative setup for acquiring coded diffraction
  patterns.}
\label{fig:xray}
\end{figure}
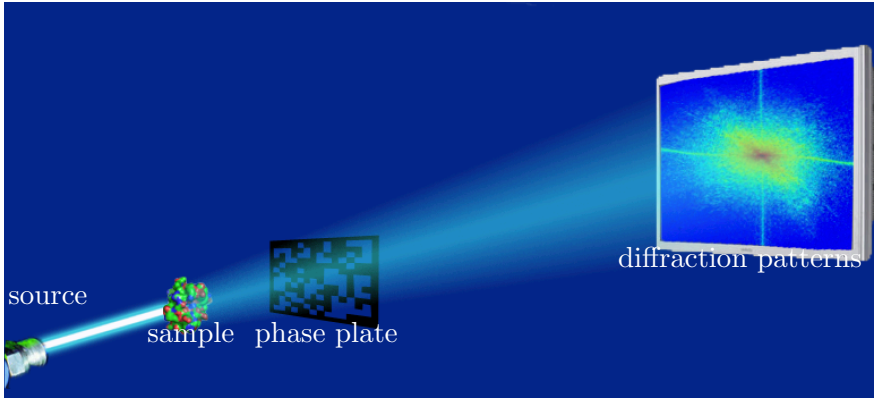

In this paper, we analyze such a data collection scheme in which one
uses multiple modulations. Our model for data acquisition is thus as
follows:
\begin{equation}
  \label{eq:model}
  y_{\ell,k} = \left| \sum_{t = 0}^{n-1} x[t] \bar{d}_\ell[t] e^{-i2\pi k t/n}
  \right|^2, \quad \begin{array}{l} 0 \le k \le n-1\\
    1 \le \ell \le L
  \end{array}. 
  \end{equation}
  In words, we collect the magnitude of the discrete Fourier transform
  (DFT) of $L$ modulations of the signal $\x$. In matrix notation,
  letting $\mtx{D}_\ell$ be the diagonal matrix with the modulation
  pattern $d_\ell[t]$ on the diagonal and $\vct{f}_k^*$ be the rows of
  the DFT, we observe
\[
y_{\ell,k} = |\vct{f}_k^* \mtx{D}^*_{\ell} \vct{x}|^2.
\]
We prove that if we use random modulation patterns (random waveforms
$d[t]$), then the solution to \eqref{eq:PL} is exact with high
probability provided that we have sufficiently many CDPs. In fact, we
will see that the feasible set in \eqref{eq:PL} equal to
\begin{equation}
  \label{eq:feasible}
  \{\mtx{X} : \mtx{X} \succeq \mtx{0} \text{ and } 
\mathcal{A}(\mtx{X}) = \vct{y}\}  
\end{equation}
reduces to a single point $\vct{x}\vct{x}^*$. Above $\mathcal{A}:
\mathcal{S}^{n\times n}\rightarrow \mathbb{R}^{m=nL}$
($\mathcal{S}^{n\times n}$ is the space of self-adjoint matrices) is
the linear mapping giving us the linear equalities in \eqref{eq:PL},
\[
\mathcal{A}(\mtx{X}) =
\left\{\vct{f}_k^*\mtx{D}_\ell^*\mtx{X}\mtx{D}_\ell\vct{f}_k\right\}_{\ell,
  k} =
\left\{\tr(\mtx{D}_\ell\vct{f}_k\vct{f}_k^*\mtx{D}_\ell^*
  \mtx{X})\right\}_{\ell,k}.
\]

\subsection{Main result}

Our model assumes random modulations and we work with diagonal
matrices $\mtx{D}_\ell$, $1 \le \ell \le L$, which are i.i.d.~copies
of a matrix $\mtx{D}$, whose entries are themselves i.i.d.~copies of a
random variable $d$. Throughout, we assume that $d$ is symmetric,
obeys $|d| \le M$ as well as the moment conditions
\begin{align}
\label{momentcond}
\E d = 0, \quad \E d^2 = 0,\quad \E \abs{d}^4 =2\E \abs{d}^2.
\end{align}
A random variable obeying these assumptions is said to be {\em
  admissible}.  The reason why we can have $\E d^2 = 0$ while $d \neq
0$ is that $d$ is complex valued. An example of an admissible random
variable is $d = b_1 b_2$, where $b_1$ and $b_2$ are independent and
distributed as
\begin{align}
\label{eq:octanary}
b_1=\begin{cases}
 1&\text{with prob.}\quad \frac{1}{4}\\
-1&\text{with prob.}\quad \frac{1}{4}\\
-i&\text{with prob.}\quad \frac{1}{4}\\
 i&\text{with prob.}\quad \frac{1}{4}
\end{cases}
\quad
\text{and}
\quad
b_2=\begin{cases}
1&\text{with prob.}\quad \frac{4}{5}\\
\sqrt{6}&\text{with prob.}\quad \frac{1}{5}
\end{cases}.
\end{align}
We would like to emphasize that we impose $\E[d^2]=0$ mostly to
simplify our exposition. In fact, the conclusion of Theorem
\ref{mainthm} below remains valid if $\E[d^2] \neq 0$, although we do
not prove this in this paper. In particular, we can also work with $d$
distributed as
\begin{equation}
\label{eq:ternary}
d=\begin{cases}
1&\quad\text{with prob.}\quad \frac{1}{4}\\
0&\quad\text{with prob.}\quad \frac{1}{2}\\
-1&\quad\text{with prob.}\quad \frac{1}{4}
\end{cases}. 
\end{equation}

\begin{theorem} 
\label{mainthm}
Suppose that the modulation is admissible and that the number $L$ of
coded diffraction patterns obeys
\begin{equation*}
  L\ge c \cdot \log^4 n,
\end{equation*}
for some fixed numerical constant $c$. Then with probability at least
$1-{1}/{n}$, the feasibility problem \eqref{eq:feasible} reduces to a
unique point, namely, $\vct{x}\vct{x}^*$, and thus recovers $\vct{x}$
up to a global phase shift.  For $\gamma \ge 1$, setting $L\ge c
  \gamma \log^4 n$ leads to a probability of success at least
  $1-n^{-\gamma}$.
\end{theorem}
Thus, in a stylized physical setting, it is possible to recover an
arbitrary signal from a fairly limited number of coded diffraction
patterns by solving an SDP feasibility problem. As mentioned earlier,
the equivalence from \cite{waldspurger2012phase} implies that our
theoretical guarantees automatically carry over to the PhaseCut
formulation.

Mathematically, the phase recovery problem is different than that in
which the sampling vectors are Gaussian as in
\cite{candes2013phase}. The reason is that the measurements in Theorem
\ref{mainthm} are far more structured and far `less random'. Loosely
speaking, our random modulation model uses on the order of $m := nL$
random bits whereas the Gaussian model with the same number of
quadratic equations would use on the order of $m n$ random bits (this
can be formalized by using the notion of entropy from information
theory). A consequence of this difference is that the proof of the
theorem requires new techniques and ideas. Having said this, an open
and interesting research direction is to close the gap---remove the
log factors---and show whether or not perfect recovery can be achieved
from a number of coded diffraction patterns independent of dimension.

The first version of this paper was made publicly available at the
same time as \cite{gross2013partial}, which begins to study the
performance of PhaseLift from non-Gaussian sampling vectors. There,
the authors study sampling schemes from certain finite vector
configurations, dubbed t-designs. These models are different from ours
and do not include our coded diffraction patterns as a special
case. Hence, our results are not comparable. Having said this, there
are similarities in the proof techniques, especially in the role
played by the robust injectivity property, compare our Lemma
\ref{lem:PL} from Section \ref{sec:robust} with \cite[Section
3.3]{gross2013partial}.
    
%

\subsection{Other approaches to phase retrieval and related works}

There are of course other approaches to phase retrieval and we mention
some literature for completeness and to inform the interested reader
of recent progress in this area.  Balan \cite{balan2010signal} studies
a problem where the sampling vectors model a short-time Fourier
transform.  Balan, Casazza and Edidin \cite{balan2007equivalence}
formulate the phase retrieval problem as nonconvex optimization
problem.  In \cite{balan2006signal}, the same authors
\cite{balan2006signal} describe some applications of the phase problem
in signal processing and speech analysis and presents some necessary
and sufficient conditions which guarantee that the solution to
\eqref{eq:general} is unique. Other articles studying the minimal
number of frame coefficient magnitudes for noiseless recovery include
\cite{cahillphase, alexeev2012phase, bandeira2013saving,
  mondragon2013determination,balan2013stability}. We recommend the two
blog posts \cite{blog1} and \cite{blog2} by Mixon and the references
therein for a comprehensive review and discussion of such
results. Lower bounds on the performance of any recovery method from
noisy data are studied in \cite{bandeira2013saving,
  balan2013invertibility, eldar2012phase}.

On the algorithmic side, \cite{balan2009nonlinear} proposes a
nonlinear scheme for phase retrieval having exponential time
complexity in the dimension of the signal $\vct{x}$ while
\cite{balan2009painless} presents a tractable algorithm requiring a
number of measurements at least quadratic in the dimension of the
signal; that is to say, $m\ge c \cdot n^2$ for some constant $c > 0$.

We also wish to mention some recent works aiming at presenting
efficient reconstruction algorithms for generic frames (for certain
types of sampling vectors such as those in \cite{alexeev2012phase,
  bandeira2013phase, raz2013vectorial} fast implementations already
exist) and give two references. The first
\cite{balan2012reconstruction} introduces an iterative regularized
least-square algorithm and establishes convergence of the
algorithm. It is however not known whether this algorithm enjoys
accurate reconstruction guarantees. The second
\cite{netrapalli2013phase} is recent and studies an alternative
minimization approach for phase retrieval, which yields very accurate
but not exact reconstructions from a number of measurements that needs
to be at least on the order of $n\log n^3$ Gaussian measurements. 

There also is a recent body of work studying the phase retrieval under
sparsity assumptions about the signal we wish to recover, see
\cite{shechtman2011sparsity, ohlsson2011compressive, li2012sparse,
  jaganathan2012robust,oymak2012simultaneously} as well as the
references therein. Finally, a different line of work
\cite{alexeev2012phase, bandeira2013phase} studies the phase retrieval
by polarization, see also \cite{raz2013vectorial} for a related
approach. This technique comes with an algorithm that can achieve
recovery using on the order of $\log n$ specially constructed
masks/codes in the noiseless case. However, to the extent of our
knowledge, PhaseLift offers more flexibility in terms of the number
and types of masks that can be used since it can be applied regardless
of the data acquisition scheme. In addition, when dealing with noisy
data PhaseLift behaves very well, see Section \ref{sec:numerical}
below and the experiments in \cite{bandeira2013phase}. 



\section{Numerical Experiments}
\label{sec:numerical}

In this section we carry out some simple numerical experiments to show
how the performance of the algorithm depends on the number of
measurements/masks and how the algorithm is affected by noise. To
solve the optimization problems below we use Auslender and Teboulle
\cite{auslender2006interior} sub-gradient optimization method with a
solver written in the framework provided by TFOCS
\cite{becker2011templates} (The code is available online at
\cite{MSweb}). The stopping criterion is when the Frobenius norm of
the relative error of the objective between two subsequent iterations
falls below $10^{-10}$ or the number of iterations reaches $50,000$,
whichever occurs first. Before presenting the results we introduce the
signal and measurement models we use.

\subsection{Signal models}
\label{sigmodels}
We consider two signal models:
\begin{itemize}
\item \emph{Random low-pass signals.} Here, $\x$ is given by 
  \begin{align*} {x}[t]=\sum_{k=-(M/2-1)}^{M/2} (X_k+iY_k) e^{2\pi i
      (k-1)(t-1)/n},
\end{align*}
with $M = n/8$ and $X_k$ and $Y_k$ are i.i.d.~$\mathcal{N}(0,1)$.

\item \emph{Random Gaussian signals.} In this model, $\x \in \C^n$ is
  a random complex Gaussian vector with i.i.d.~entries of the form
  $x[t] = X+iY$ with $X$ and $Y$ distributed as $\mathcal{N}(0,1)$;
  this can be expressed as
  \begin{align*} {x}[t]=\sum_{k=-(n/2-1)}^{n/2} (X_k+iY_k) e^{2\pi i
      (k-1)(t-1)/n},
\end{align*}
where $X_k$ and $Y_k$ are are i.i.d.~$\mathcal{N}(0,1/8)$ so that the
low-pass model is a `bandlimited' version of this high-pass random
model (variances are adjusted so that the expected power is the same).
\end{itemize}

\subsection{Measurement models}
\label{measurementmodels}
We perform simulations based on four different kinds of measurements:

\begin{itemize}
\item \emph{Gaussian measurements.} We sample $m=nL$ random complex
  Gaussian vectors $\vct{a}_k$ and use measurements of the form
  $|\vct{a}_k^*\vct{x}|^2$.

\item \emph{Binary modulations/codes.} We sample ($L-1$) binary codes
  distributed as
\begin{equation*}
d=\begin{cases}
1&\quad\text{with prob.}\quad \frac{1}{2}\\
0&\quad\text{with prob.}\quad \frac{1}{2}
\end{cases}
\end{equation*}
together with a regular diffraction pattern ($d[t] = 1$ for all
$t$). 

\item \emph{Ternary modulations/codes.} We sample ($L-1$) ternary
  codes distributed as \eqref{eq:ternary} 
together with a regular diffraction pattern.

\item \emph{Octanary modulations/codes.} Here, the codes are
  distributed as \eqref{eq:octanary}.
\end{itemize}

\subsection{Phase transitions}

We carry out some numerical experiments to show how the performance of
the algorithm depends on the number of measurements/coded
patterns. For this purpose we consider $50$ trials. In each trial we
generate a random complex vector $\vct{x}\in\C^n$ (with $n=128$) from
both signal models and gather data according to the four different
measurement models above.  For each trial we solve the following
optimization problem
\begin{equation}
\label{approxSDP}
    \min \quad \frac{1}{2}\twonorm{\vct{b}-\mathcal{A}(\mtx{X})}^2+\lambda\text{tr}(\mtx{X})\quad\text{subject to}\quad\mtx{X}\succeq \mtx{0}\\
\end{equation}
with $\lambda=10^{-3}$ (Note that the solution to \eqref{approxSDP} as $\lambda$ tends to zero will equal to the optimal solution of \eqref{eq:PL}).  

In Figure \ref{PTs} we report the empirical probability of success for
different signal and measurement models with different number of
measurements. We declare a trial successful if the relative error of
the reconstruction
($\fronorm{\hat{\mtx{X}}-\vct{x}\vct{x}^*}/\fronorm{\vct{x}\vct{x}^*}$)
falls below $10^{-5}$). These plots suggest that for the type of
models studied in this paper six coded patterns are sufficient for
exact recovery via convex programming.

\begin{figure}[h]
        \centering
\begin{tikzpicture}[scale=1] 
\begin{groupplot}[group style={group size=2 by 1,horizontal sep=1cm,xlabels at=edge bottom, ylabels at=edge left,xticklabels at=edge bottom},xlabel=L,
        ylabel=Probability of success,
        legend pos= south east]
 \nextgroupplot[title={random Gaussian signals}]
        
        \addplot +[mark=*,solid,blue,line width=1pt] table[x index=0,y index=2]{./prob};
        \addlegendentry{binary masks}
        \addplot +[mark=square ,solid,teal,line width=1pt] table[x index=0,y index=1]{./prob};
        \addlegendentry{ternary masks}
        \addplot +[mark=triangle ,solid,red,line width=1pt] table[x index=0,y index=3]{./prob};
        \addlegendentry{octanary masks}
        \addplot +[mark=diamond ,solid,orange,line width=1pt] table[x index=0,y index=4]{./prob};\addlegendentry{Gaussian meas.}

 \nextgroupplot[title={random low-pass signals}]
        \addplot +[mark=*,solid,blue,line width=1pt] table[x index=0,y index=2]{./probsmooth};\addlegendentry{binary masks}
        \addplot +[mark=square ,solid,teal,line width=1pt] table[x index=0,y index=1]{./probsmooth};\addlegendentry{ternary masks}
        \addplot +[mark=triangle ,solid,red,line width=1pt] table[x index=0,y index=3]{./probsmooth};\addlegendentry{octanary masks}
        \addplot +[mark=diamond ,solid,orange,line width=1pt] table[x index=0,y index=4]{./probsmooth};\addlegendentry{Gaussian meas.}
 \end{groupplot}
\end{tikzpicture}
\caption{ Empirical probability of success based on $50$ random trials
  for different signal/measurement models and a varied number of
  measurements. A value of $L$ on the x-axis means that we have a
  total of $m=Ln$ samples.}
\label{PTs}
\end{figure}
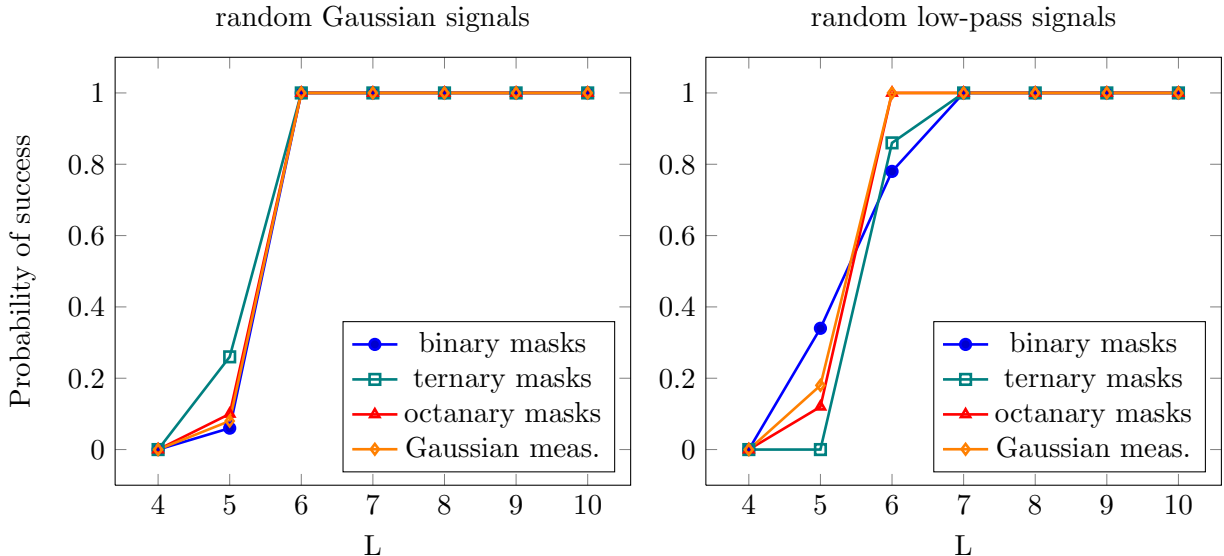

\subsection{Noisy measurements}

We now study how the performance of the algorithm behaves in the
presence of noise. We consider Poisson noise which is the usual noise
model in optics. More, specifically we assume that the measurements
$\{y_k\}_{k=1}^m$ is a sequence of independent samples from the
Poisson distributions Poi$(\mu_k)$, where $\mu_k=|\vct{a}_k^*\x|^2$
correspond to the noiseless measurements. The Poisson log-likelihood
for independent samples has the form $\sum_k y_k\log \mu_k-\mu_k$ (up
to an additive constant factor). Following a classical fitting
approach we balance a maximum likelihood term with the trace norm in
the relaxation \eqref{eq:PL}:
 \begin{align*}
   \min\sum_k[\mu_k-y_k\log
   \mu_k]+\lambda\text{tr}(\mtx{X})\quad\text{subject
     to}\quad\vct{\mu}=\mathcal{A}(\mtx{X})\quad\text{and}\quad\mtx{X}\succeq\mtx{0}.
 \end{align*}
 The test signal is again a complex random signal sampled according to
 the two models described in Section \ref{sigmodels}. We use eight
 CDP's according to the three models described in Section
 \ref{measurementmodels}. Poisson noise is
 adjusted so that the SNR levels range from $10$ to $50$dB. Here,
 $\text{SNR}
 =\twonorm{\mathcal{A}(\x\x^*)}/\twonorm{\vct{b}-\mathcal{A}(\x\x^*)}$
 is the signal-to-noise ratio. For the regularization parameter we use
 $\lambda=1/\text{SNR}$. (In these experiments, the value of SNR is
 known. The result, however, is rather insensitive to the choice of
 the parameter $\lambda$ and a good choice for the regularization
 parameter $\lambda$ can be obtained by cross validation.) For each
 SNR level we repeat the experiment ten times with different random
 noise and different random CDP's.

 Figure \ref{SNRs} shows the average relative MSE (in dB) versus the
 SNR (also in dB). More precisely, the values of
 $10\log_{10}(\text{rel.~MSE})$ are plotted, where $\text{rel.~MSE} =
 \fronorm{\hat{\mtx{X}}-\x\x^*}^2/\fronorm{\hat{\mtx{X}}}^2$. These
 figures indicate that the performance of the algorithm degrades
 linearly as the SNR decreases (on a dB/dB scale). Empirically, the
 slope is close to -1, which means that the MSE scales like the
 noise. Together with the low offset, these features indicate that all
 is as in a well-conditioned-least squares problem.
 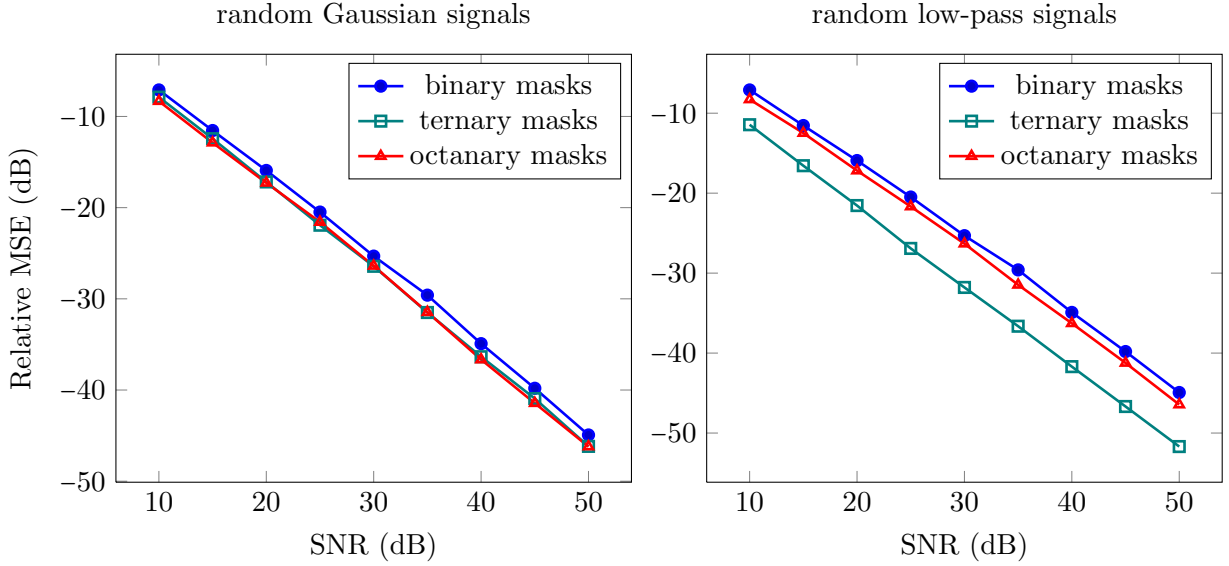
\begin{figure}[h]
        \centering
\begin{tikzpicture}[scale=1] 
\begin{groupplot}[group style={group size=2 by 1,horizontal sep=1cm,xlabels at=edge bottom, ylabels at=edge left,xticklabels at=edge bottom},xlabel=SNR (dB),
        ylabel=Relative MSE (dB)]
 \nextgroupplot[title={random Gaussian signals}]
        
        \addplot +[mark=*,solid,blue,line width=1pt] table[x index=0,y index=1]{./SNR};\addlegendentry{binary masks}
        \addplot +[mark=square ,solid,teal,line width=1pt] table[x index=0,y index=2]{./SNR};\addlegendentry{ternary masks}
        \addplot +[mark=triangle ,solid,red,line width=1pt] table[x index=0,y index=3]{./SNR};\addlegendentry{octanary masks}
        
        \nextgroupplot[title={random low-pass signals}]
        \addplot +[mark=*,solid,blue,line width=1pt] table[x index=0,y index=1]{./SNR};\addlegendentry{binary masks}
        \addplot +[mark=square ,solid,teal,line width=1pt] table[x index=0,y index=2]{./SNRsmooth};\addlegendentry{ternary masks}
        \addplot +[mark=triangle ,solid,red,line width=1pt] table[x index=0,y index=3]{./SNRsmooth};\addlegendentry{octanary masks}
\end{groupplot}
\end{tikzpicture}
\caption{SNR versus relative MSE on a dB-scale for different kinds of
  signal/measurement models. The linear relationship between SNR and
  MSE (on the dB scale) is apparent. The MSE behaves as in a
    well-conditioned least-squares problem.}
\label{SNRs}
\end{figure}

\section{Proofs}

We prove our results in this section. Before we begin, we introduce
some notation.  We recall that the random variable $d$ is admissible,
i.e. bounded i.e.~$|d|\le M$, symmetric, and obeying moment
constraints
\begin{align}
  \E d =0,
  \quad\E d^2 =0,\quad\E \abs{d}^4 =2\E \abs{d}^2.
\end{align}
Without loss of generality we also assume that $\E |d|^2 = 1$.
Throughout $\mtx{D}$ is a diagonal matrix with i.i.d.~entries
distributed as $d$. 
For a vector $\vct{y}\in\C^n$ we use
$\vct{y}^T$ and $\vct{y}^*$ to denote the transpose and complex
conjugate of the vector $\vct{y}$. We also use $\bar{\vct{y}}$ to
denote elementwise conjugation of the entries of $\vct{y}$. Since this
is less standard, we prefer to be concrete as to avoid ambiguity: for
example,
\begin{align*}
\begin{bmatrix}1+i \\1+2i\end{bmatrix}^T=\begin{bmatrix}1+i &1+2i\end{bmatrix},\quad\begin{bmatrix}1+i \\1+2i\end{bmatrix}^*=\begin{bmatrix}1-i &1-2i\end{bmatrix},\quad\overline{\begin{bmatrix}1+i \\1+2i\end{bmatrix}}=\begin{bmatrix}1-i \\1-2i\end{bmatrix}.
\end{align*}
Continuing, $\|\mtx{X}\|$ is the spectral or operator norm of a matrix
$\mtx{X}$. Finally, $\vct{1}$ is a vector with all entries equal to
one.

Throughout, we assume that the fixed vector $\x$ we seek to recover is
unit normed, i.e.~$\twonorm{\vct{x}}=1$. Throughout $T$ is the linear
subspace
\[
T = \{\mtx{X} = \vct{x}\vct{y}^*+\vct{y}\vct{x}^* \, : \,
\vct{y}\in\C^n\}.
\]
This subspace may be interpreted as the tangent space at
$\vct{x}\vct{x}^*$ to the manifold of Hermitian matrices of rank
$1$. Below $T^\perp$ is the orthogonal complement to $T$.  For a
linear subspace $V$ of Hermitian matrices, we use $\mtx{Y}_V$ or
$\mathcal{P}_V(\mtx{Y})$ to denote the orthogonal projection of
$\mtx{Y}$ onto $V$. With this, the reader will check that
$\mtx{Y}_{T^\perp} = (\mtx{I} - \vct{x}\vct{x}^*) \mtx{Y} (\mtx{I} -
\vct{x}\vct{x}^*)$.

\subsection{Preliminaries} 

It is useful to record two identities that shall be used multiple
times, and defer the proofs to the Appendix.
\begin{lemma}
\label{teo:expectation1}
For any fixed vector $\vct{x} \in \mathbb{C}^n$
\[
\E\left(\frac{1}{nL}\mathcal{A}^*\mathcal{A}(\vct{x}\vct{x}^*)\right)=\E\left(\frac{1}{n}\sum_{k=1}^n \left|\vct{f}_k^*\mtx{D}^*\vct{x}\right|^2\mtx{D}\vct{f}_k\vct{f}_k^*\mtx{D}^*\right)=\vct{x}\vct{x}^*+\twonorm{\vct{x}}^2\mtx{I}.
\]
\end{lemma}

\begin{lemma}
\label{teo:expectation2}
For any fixed $\vct{x} \in \mathbb{C}^n$, 
\begin{equation*}
  \E\left(\frac{1}{n}\sum_{k=1}^n{\big(\vct{f}_k^*\mtx{D}^*\vct{x}\big)}^2\mtx{D}\vct{f}_k\vct{f}_k^T\mtx{D}\right)=2\vct{x}\vct{x}^T.
\end{equation*}
\end{lemma}
Next, we present two simple intermediate results we shall also
use. The proofs are also in the Appendix.
\begin{lemma}
\label{identityinrange}
Fix $\delta > 0$ and suppose the number $L$ of CDP's obeys $L \ge c
\log n$ for some sufficiently large numerical constant $c$. Then with
probability at least $1-1/n^2$,
\begin{align*}
  \left\|\frac{1}{nL}\mathcal{A}^*(\vct{1})-\mtx{I}_n\right\|\le
  \delta.
\end{align*}
\end{lemma}

\begin{lemma}
\label{upperRIP1}
For all positive semidefinite matrices $\mtx{X}$, it holds
\begin{align*}
\frac{1}{nL}\onenorm{\mathcal{A}(\mtx{X})}\le M^2\tr(\mtx{X}). 
\end{align*}
\end{lemma}

Finally, the last piece of mathematics is the matrix Hoeffding
inequality
\begin{lemma}\cite[Theorem 1.3]{tropp2012user}
\label{mathoef}
Let $\{S_\ell\}_{\ell=1}^L$ be a sequence of independent random $n
\times n$ self-adjoint matrices. Assume that each random matrix obeys
\begin{align}
\label{condhoeffding}
\E \mtx{S}_\ell=\mtx{0}\quad\text{and}\quad \|\mtx{S}_\ell\|\preceq
\Delta\quad\text{almost surely}.
\end{align}
Then for all $t\ge0$,
\begin{equation}
\label{eq:hoeffding}
\mathbb{P}\Bigl(\frac{1}{L}\|\sum_{\ell=1}^L\mtx{S}_\ell\|\ge
t\Bigr)\le 2n \exp\Bigl(-\frac{Lt^2}{8\Delta^2}\Bigr).
\end{equation}
\end{lemma}

\subsection{Certificates}

We now establish sufficient conditions guaranteeing that
$\vct{x}^*\vct{x}$ is the unique feasible point of
\eqref{eq:feasible}. Variants of the lemma below have appeared before
in the literature, see
\cite{candes2012phaselift,demanet2012stable,candes2012solving}.
\begin{lemma}
\label{lem:ineact}
Suppose the mapping $\mathcal{A}$ obeys the following two
properties: 
\begin{enumerate}
\item For all matrices $\mtx{X}\in T$
\begin{equation}
\label{local2RIP}
\frac{1}{\sqrt{nL}} \twonorm{\cA(\mtx{X})} \ge \frac{(1-\delta)}{\sqrt{2}} \|\mtx{X}\|_F.
\end{equation}
\item There exists a self-adjoint matrix of the form
  $\mtx{Z}=\cA^*(\vct{\lambda})$, with $\vct{\lambda}$ real valued
  (this makes sure that $\mtx{Z}$ is self adjoint), obeying
\begin{equation}
\label{approxcert}
\mtx{Z}_{T^\perp}\preceq -\mtx{I}_{T^\perp}\quad\text{and}\quad\|\mtx{Z}_T\|_F\leq \frac{1-\delta}{2M^2 \sqrt{nL}}.
\end{equation}
\end{enumerate}
Then $\vct{x}^*\vct{x}$ is the unique element in the feasible set
\eqref{eq:feasible}.
\end{lemma}
\begin{proof}
  Suppose $\vct{x}\vct{x}^* + \mtx{H}$ is feasible. Feasibility
  implies that $\mtx{H}$ is a self-adjoint matrix in the null space of
  $\cA$ and $\mtx{H}_{T^\perp} \succeq \mtx{0}$. This is because for
  all $\vct{y} \perp \vct{x}$, 
\[
\vct{y}^*(\vct{x}\vct{x}^* + \mtx{H}) \vct{y} = \vct{y}^* \mtx{H}
\vct{y} \ge 0, 
\]
which says that $\mtx{H}_{T^\perp}$ is positive semidefinite. This
gives
\[
\<\mtx{H}, \mtx{Y}\> = 0 = \<\mtx{H}_T, \mtx{Z}_T\> +
\<\mtx{H}_{T^\perp}, \mtx{Z}_{T^\perp}\>.
\] 
On the one hand,
\begin{equation}
\label{Hequality}
\langle \mtx{H}_T, \mtx{Z}_T \rangle=-\langle \mtx{H}_{T^\perp},
\mtx{Z}_{T^\perp}\rangle\geq \langle \mtx{H}_{T^\perp},
\mtx{I}_{T^\perp}\rangle= \tr(\mtx{H}_{T^\perp}).
\end{equation}
Therefore,
\[
\tr(\mtx{H}_{T^\perp})\geq\frac{1}{M^2nL} \onenorm{\cA(\mtx{H}_{T^\perp})}\geq\frac{1}{M^2nL} \twonorm{\cA(\mtx{H}_{T^\perp})}.
\]
where the first inequality above follows from Lemma
\ref{upperRIP1}. The injectivity property \eqref{local2RIP} gives
 \[
 \frac{1}{\sqrt{nL}}\twonorm{\cA(\mtx{H}_{T})} \geq \frac{(1-\delta)}{\sqrt{2}} \|\mtx{H}_T\|_F
 \]
 and since $\cA(\mtx{H}_{T}) = -\cA(\mtx{H}_{T^\perp})$, we
 established
\begin{equation}
\label{eq:one}
\tr(\mtx{H}_{T^\perp})\geq \frac{1-\delta}{\sqrt{2nL}M^2} \|\mtx{H}_T\|_F.
\end{equation}
On the other hand,
\begin{equation}
\label{eq:two}
|\langle \mtx{H}_T, \mtx{Z}_T \rangle|\leq
\|\mtx{H}_T\|_F\|\mtx{Z}_T\|_F\leq \frac{1-\delta}{2M^2 \sqrt{nL}}\|\mtx{H}_{T}\|_F.
\end{equation}
In summary, \eqref{Hequality}, \eqref{eq:one} and \eqref{eq:two}
assert that $\mtx{H}_T = 0$. In turn, this gives
$\tr(\mtx{H}_{T^\perp}) = 0$ by \eqref{Hequality}, which implies that
$\mtx{H}_{T^\perp}=\mtx{0}$ since $\mtx{H}_{T^\perp} \succeq
\mtx{0}$. This completes the proof.
\end{proof}

Property \eqref{local2RIP} can be viewed as a form of robust
injectivity of the mapping $\mathcal{A}$ restricted to elements in
$T$.  It is of course reminiscent of the local restricted isometry
property in compressive sensing. Property \eqref{approxcert} can be
interpreted as the existence of an approximate dual certificate. It is
well known that injectivity together with an exact dual certificate
leads to exact reconstruction. The above lemma essentially asserts
that a robust form of injectivity together with an approximate dual
certificate leads to exact recovery as in \cite[Section
2.1]{candes2012phaselift}, see also \cite{gross2011recovering}.  In
the next two sections we show that the two properties stated in Lemma
\ref{lem:ineact} above each hold with probability at least $1-1/(2n)$.

\subsection{Robust injectivity}
\label{sec:robust}

\begin{lemma}
\label{lem:PL}
Fix $\delta > 0$ and suppose $L$ obeys $L \ge c \log^3 n$ for some
sufficiently large numerical constant $c$. Then with probability at
least $1-1/2n$, for all $\mtx{X}\in T$,
\[
\frac{1}{\sqrt{nL}} \twonorm{\cA(\mtx{X})} \ge \frac{(1-\delta)}{\sqrt{2}}
\|\mtx{X}\|_F.
\]
\end{lemma}
\begin{proof}
  First, notice that without loss of generality we can assume that
  $\vct{x}^*\vct{y}$ is real valued in the definition of $T$. That is,
\[
T=\{\mtx{X} = \vct{x}\vct{y}^*+\vct{y}\vct{x}^* \, : y \in \C^n \text{
  and } \vct{x}^*\vct{y} \in \R\}.
\]
The reason why this is true is that for any $\vct{y} \in \C^n$, we can
find $\lambda \in \mathbb{R}$, such that
$\vct{x}^*\vct{y}-i\lambda\vct{x}^*\vct{x}=\vct{x}^*(\vct{y}-i\lambda\vct{x})\in
\mathbb{R}$ while
\[
\vct{x}(\vct{y}-i\lambda\vct{x})^*+(\vct{y}-i\lambda\vct{x})\vct{x}^*=\vct{x}\vct{y}^*+\vct{y}\vct{x}^*,
\]

Now for any $\mtx{X} = \x\y^* + \y\x^* \in T$, 
\begin{align}
\label{twofroineq}
\fronorm{\mtx{X}} = \fronorm{\x\y^* + \y\x^*} \le \fronorm{\x\y^*} +
\fronorm{\y\x^*} \le 2 \twonorm{\x} \twonorm{\y} = 2 \twonorm{\y},
\end{align}
where we recall that $\twonorm{\x} = 1$. Hence, it suffices to show that 
\begin{equation}
  \label{eq:toshow}
  \frac{1}{\sqrt{nL}}  \twonorm{\mathcal{A}(\vct{x}\vct{y}^*+\vct{y}\vct{x}^*)} \ge \frac{(1-\delta)}{\sqrt{2}}
  \twonorm{\y}.
\end{equation}
We have
\[
\twonorm{\mathcal{A}(\vct{x}\vct{y}^*+\vct{y}\vct{x}^*)}^2 =
\sum_{\ell=1}^L\sum_{k=1}^n{\bigg(\vct{f}_k^*\mtx{D}_\ell^*(\vct{x}\vct{y}^*+\vct{y}\vct{x}^*)\mtx{D}_\ell\vct{f}_k\bigg)}^2. 
\]
(The reader might have expected a sum of squared moduli but since
$\vct{x}\vct{y}^*+\vct{y}\vct{x}^*$ is self adjoint,
$\vct{f}_k^*\mtx{D}_\ell^*(\vct{x}\vct{y}^*+
\vct{y}\vct{x}^*)\mtx{D}_\ell\vct{f}_k$ is real valued and so we can
just as well use squares.
For exposition purposes, set 
\[
\mtx{A}_k(\mtx{D})=\abs{\vct{f}_k^*\mtx{D}^*\vct{x}}^2\vct{f}_k\vct{f}_k^*,
\quad
\mtx{B}_k(\mtx{D})={(\vct{f}_k^*\mtx{D}^*\vct{x})}^2\vct{f}_k\vct{f}_k^T.
\]
A simple computation we omit yields
\begin{align*}
\bigg(\vct{f}_k^*\mtx{D}^*\vct{x}\vct{y}^*\mtx{D}\vct{f}_k+\vct{f}_k^*\mtx{D}^*\vct{y}\vct{x}^*\mtx{D}\vct{f}_k\bigg)^2
& = \begin{bmatrix} \vct{y} \\ \bar{\vct{y}} \end{bmatrix}^*
\begin{bmatrix} \mtx{D} & \mtx{0} \\ \mtx{0} & \mtx{D}^* \end{bmatrix}
\begin{bmatrix} 
\mtx{A}_k(\mtx{D})
& 
\mtx{B}_k(\mtx{D})
\\ 
\overline{\mtx{B}_k(\mtx{D})}
& 
\overline{\mtx{A}_k(\mtx{D})}
\end{bmatrix} 
\begin{bmatrix} \mtx{D}^* & \mtx{0} \\ \mtx{0} & \mtx{D} \end{bmatrix}
\begin{bmatrix} \vct{y} \\ \bar{\vct{y}} \end{bmatrix}\\
& = \begin{bmatrix} \vct{y} \\ \bar{\vct{y}} \end{bmatrix}^* \mtx{W}_k(\mtx{D}) \begin{bmatrix} \vct{y} \\ \bar{\vct{y}} \end{bmatrix}, 
\end{align*}
where 
\[
\mtx{W}_k(\mtx{D}) := \begin{bmatrix} \mtx{D} & \mtx{0} \\ \mtx{0} & \mtx{D}^* \end{bmatrix}
\begin{bmatrix} 
\mtx{A}_k(\mtx{D})
& 
\mtx{B}_k(\mtx{D})
\\ 
\overline{\mtx{B}_k(\mtx{D})}
& 
\overline{\mtx{A}_k(\mtx{D})}
\end{bmatrix} 
\begin{bmatrix} \mtx{D}^* & \mtx{0} \\ \mtx{0} & \mtx{D} \end{bmatrix}. 
\]
Fix a positive threshold $T_n$. We now claim that \eqref{eq:toshow}
follows from
\begin{equation}
\label{eq:toshow2}
\frac{1}{nL} \sum_{\ell = 1}^L \sum_{k = 1}^n \mtx{W}_k(\mtx{D}_\ell)
1(\left|\vct{f}_k^*\mtx{D}_\ell^*\vct{x}\right|\leq T_n)
\succeq \alpha \begin{bmatrix} \vct{x} \\
  -\overline{\vct{x}} \end{bmatrix}\begin{bmatrix} \vct{x} \\
  -\overline{\vct{x}} \end{bmatrix}^*+{(1-\delta)}^2\mtx{I}_{2n}
\end{equation}
in which $\alpha$ is any real valued number. To see why this is true,
observe that 
\begin{align*}
  \frac{1}{nL}
  \twonorm{\mathcal{A}(\vct{x}\vct{y}^*+\vct{y}\vct{x}^*)}^2 &
  \ge \begin{bmatrix} \vct{y} \\ \bar{\vct{y}} \end{bmatrix}^*
  \frac{1}{nL} \sum_\ell \sum_k \mtx{W}_k(\mtx{D}_\ell)
  1(\left|\vct{f}_k^*\mtx{D}_\ell^*\vct{x}\right|\leq
  T_n) \begin{bmatrix} \vct{y} \\ \bar{\vct{y}} \end{bmatrix} \\
  & \ge (1-\delta)^2 \begin{bmatrix} \vct{y} \\
    \bar{\vct{y}} \end{bmatrix}^* \mtx{I}_{2n} \begin{bmatrix} \vct{y}
    \\ \bar{\vct{y}} \end{bmatrix} = 2(1-\delta)^2 \twonorm{\y}^2.
\end{align*}
The last inequality comes from \eqref{eq:toshow2} together with 
\[
\begin{bmatrix} \vct{x} \\ -\overline{\vct{x}} \end{bmatrix}^*\begin{bmatrix} \vct{y} \\ \overline{\vct{y}} \end{bmatrix}=0, 
\]
which holds since we assumed that $\x^* \y$ is real valued.

The remainder of the proof justifies \eqref{eq:toshow2} by means of
the matrix Hoeffding inequality. Let $\avg{\mtx{W}}$ be the left-hand
side in \eqref{eq:toshow2} (we use notation from physics to denote
empirical averages since a bar denotes complex conjugation and we
would like to avoid overloading symbols). By definition
$\avg{\mtx{W}}$ is the empirical average of $L$ i.i.d.~copies of
\[
\mtx{W}(\mtx{D}) = \frac{1}{n} \sum_{k=1}^n \mtx{W}_k(\mtx{D})
\mathbb{1}_{\{\left|\vct{f}_k^*\mtx{D}^*\vct{x}\right|\leq T_n\}}.
\]
First, $\mtx{W}_k(\mtx{D}) \succeq \mtx{0}$ since 
\[
\begin{bmatrix} 
\mtx{A}_k(\mtx{D})
& 
\mtx{B}_k(\mtx{D})
\\ 
\overline{\mtx{B}_k(\mtx{D})}
& 
\overline{\mtx{A}_k(\mtx{D})}
\end{bmatrix}
= \begin{bmatrix} (\vct{f}_k^*\mtx{D}^*\vct{x})\vct{f}_k \\ (\overline{\vct{f}_k^*\mtx{D}^*\vct{x}})\overline{\vct{f}_k} \end{bmatrix}\begin{bmatrix} (\vct{f}_k^*\mtx{D}^*\vct{x})\vct{f}_k \\ (\overline{\vct{f}_k^*\mtx{D}^*\vct{x}})\overline{\vct{f}_k} \end{bmatrix}^*. 
\]
Further, 
\begin{align*}
\begin{bmatrix} 
\mtx{A}_k(\mtx{D})
& 
\mtx{B}_k(\mtx{D})
\\ 
\overline{\mtx{B}_k(\mtx{D})}
& 
\overline{\mtx{A}_k(\mtx{D})}
\end{bmatrix} & = \begin{bmatrix} 
2\mtx{A}_k(\mtx{D})
& 
\mtx{0}
\\ 
\mtx{0} 
& 
2\overline{\mtx{A}_k(\mtx{D})}
\end{bmatrix} - 
\begin{bmatrix} 
\mtx{A}_k(\mtx{D})
& 
-\mtx{B}_k(\mtx{D})
\\ 
-\overline{\mtx{B}_k(\mtx{D})}
& 
\overline{\mtx{A}_k(\mtx{D})}
\end{bmatrix}\\
& \preceq \begin{bmatrix} 
2\mtx{A}_k(\mtx{D})
& 
\mtx{0}
\\ 
\mtx{0} 
& 
2\overline{\mtx{A}_k(\mtx{D})}
\end{bmatrix}. 
\end{align*}
The inequality comes from 
\[
\begin{bmatrix} 
\mtx{A}_k(\mtx{D})
& 
-\mtx{B}_k(\mtx{D})
\\ 
-\overline{\mtx{B}_k(\mtx{D})}
& 
\overline{\mtx{A}_k(\mtx{D})}
\end{bmatrix} = \begin{bmatrix} (\vct{f}_k^*\mtx{D}^*\vct{x})\vct{f}_k
\\
-(\overline{\vct{f}_k^*\mtx{D}^*\vct{x}})\overline{\vct{f}_k} \end{bmatrix}\begin{bmatrix}
(\vct{f}_k^*\mtx{D}^*\vct{x})\vct{f}_k \\
-(\overline{\vct{f}_k^*\mtx{D}^*\vct{x}})\overline{\vct{f}_k} \end{bmatrix}^*
\succeq \mtx{0}.
\]
Hence,
\begin{align*}
\sum_{k=1}^n
\begin{bmatrix} 
\mtx{A}_k(\mtx{D})
& 
\mtx{B}_k(\mtx{D})
\\ 
\overline{\mtx{A}_k(\mtx{D})}
& 
\overline{\mtx{B}_k(\mtx{D})}
\end{bmatrix}
\mathbb{1}_{\{\left|\vct{f}_k^*\mtx{D}^*\vct{x}\right|\leq T_n\}} &\preceq 2
\sum_{k=1}^n \begin{bmatrix}
  \abs{f_k^*\mtx{D}^*\vct{x}}^2\vct{f}_k\vct{f}_k^* & \mtx{0}
  \\
  \mtx{0} &
  \abs{f_k^*\mtx{D}^*\vct{x}}^2\overline{\vct{f}_k}\vct{f}_k^T
\end{bmatrix}
\mathbb{1}_{\{\left|\vct{f}_k^*\mtx{D}^*\vct{x}\right|\leq T_n\}}
\\
&\preceq 2T_n^2 \sum_{k=1}^n
\begin{bmatrix} 
\vct{f}_k\vct{f}_k^* 
& 
\mtx{0}
\\ 
\mtx{0}
& 
\overline{\vct{f}_k}\vct{f}_k^T
\end{bmatrix}\\
& 
= 2n T_n^2 \mtx{I}_{2n}.
\end{align*}
In summary, 
\[
\|\mtx{W}(\mtx{D})\|\leq 2 T_n^2 \|\mtx{D}\|^2 \leq 2M^2 T_n^2. 
\]

We now roughly estimate the mean of $\mtx{W}(\mtx{D})$. Obviously,
\begin{align*}
\W(\mtx{D}) &= \frac{1}{n} \sum_k \W_k(\mtx{D}) - \frac{1}{n} \sum_k \W_k(\mtx{D})
\mathbb{1}_{\{\left|\vct{f}_k^*\mtx{D}^*\vct{x}\right|\leq T_n\}}\\
&:=\tilde{\W}(\mtx{D})-\frac{1}{n} \sum_k \W_k(\mtx{D})
\mathbb{1}_{\{\left|\vct{f}_k^*\mtx{D}^*\vct{x}\right|\leq T_n\}}.
\end{align*}
By Lemmas \ref{teo:expectation1} and \ref{teo:expectation2}, the mean
of the first term ($\tilde{\W}(\mtx{D})$) is equal to
\begin{align}
\label{meanfirst}
\E \tilde{\W}(\mtx{D})=\mtx{I}_{2n} + \begin{bmatrix} \vct{x}\vct{x}^* &   2\vct{x}\vct{x}^T \\  2\bar{\vct{x}}\vct{x}^* &  \overline{\vct{x}}\vct{x}^T  \end{bmatrix}.
\end{align}
Furthermore, a simple calculation shows that since 
\[
|\vct{f}_k^*\mtx{D}^*\vct{x}| \le \twonorm{\vct{f}_k}
\twonorm{\mtx{D}^*\vct{x}} \le \twonorm{\vct{f}_k} \|\mtx{D}\| \le
\sqrt{n} \|\mtx{D}\|,
\]
one can verify that 
\[
\|\W_k(\mtx{D})\| \le 4n^2 \|\mtx{D}\|^4 \le 4 M^4 n^2. 
\]
Therefore, Jensen's inequality gives
\begin{align*}
  \left\|\E\W(\mtx{D})-\E\tilde{\W}(\mtx{D})\right\| & = \left\| \E \frac{1}{n} \sum_k \W_k(\mtx{D})
    \mathbb{1}_{\{\left|\vct{f}_k^*\mtx{D}^*\vct{x}\right|\leq T_n\}} \right\| \\ & \leq
  \E \left\| \frac{1}{n} \sum_k \W_k(\mtx{D})
    \mathbb{1}_{\{\left|\vct{f}_k^*\mtx{D}^*\vct{x}\right|\leq T_n\}} \right\| \\ & \leq
  4M^4n\sum_{k=1}^n\P(\left|\vct{f}_k^*\mtx{D}^*\vct{x}\right|>T_n). 
\end{align*}
Setting $T_n = \sqrt{2\beta \log n}$, then a simple application of
Hoeffding's inequality gives 
\[
\P(\left|\vct{f}_k^*\mtx{D}^*\vct{x}\right|> \sqrt{2\beta \log n}) \leq
2n^{-\beta}
\]
(we omit the details). Therefore,
\begin{align}
\label{jensen}
\left\|\E\W(\mtx{D})-\E\tilde{\W}(\mtx{D})\right\|\le \frac{8M^4}{n^{\beta-2}}. 
\end{align}
Next,
\[
\|\mtx{W}(\mtx{D})-\E \mtx{W}(\mtx{D})\| \leq \|\mtx{W}(\mtx{D})\| + \|\E\mtx{W}(\mtx{D})\|\leq \|\mtx{W}(\mtx{D})\| + \|\E\tilde{\mtx{W}}(D)\|+\|\E\W(\mtx{D})-\E\tilde{\W}(\mtx{D})\|
\]
and with $T_n$ as above, collecting our estimates gives
\[
\|\mtx{W}(\mtx{D})-\E \mtx{W}(\mtx{D})\| \le 4M^2 \beta \log n + 4 +
\frac{8M^4}{n^{\beta-2}} :=\Delta.
\]

We have done the groundwork to apply the matrix Hoeffding inequality
\eqref{eq:hoeffding}, which reads
\[
\P\left(\left\|\avg{\mtx{W}} - \E\mtx{W}(\mtx{D})\right\|\geq t\right)\leq
2n\exp\Bigl(-\frac{L t^2}{8\Delta^2}\Bigr).
\]
This implies that when $\beta$ is sufficiently large and $L\ge c
\log^3 n$ for a sufficiently large constant,
\[
\|\avg{\mtx{W}} -\E\mtx{W}(\mtx{D}) \|\leq {\epsilon}/{2}
\]
with probability at least $1-1/(2n)$.  Now from \eqref{meanfirst}, and
\eqref{jensen} gives
\[
\E\mtx{W}(\mtx{D}) = \mtx{I}_{2n}+\frac{3}{2}\begin{bmatrix}
      \vct{x} \\ \bar{\vct{x}} \end{bmatrix} [\vct{x}^*,
    \vct{x}^T]-\frac{1}{2}\begin{bmatrix} \vct{x} \\
      -\bar{\vct{x}} \end{bmatrix} [\vct{x}^*,
    -\vct{x}^T] + \mtx{E}, 
\]
where \eqref{jensen} gives that $\|\mtx{E}\| \le \epsilon/2$ provided
$\beta\ge 2 + \log (16 M^4 \epsilon^{-1})/\log n$.  Hence, we have
established that
\[
\avg{\W} \succeq (1-\epsilon) \mtx{I}_{2n} +
\frac{3}{2}\begin{bmatrix} \vct{x} \\ \bar{\vct{x}} \end{bmatrix}
[\vct{x}^*,
\vct{x}^T]-\frac{1}{2}\begin{bmatrix} \vct{x} \\
  -\bar{\vct{x}} \end{bmatrix} [\vct{x}^*, -\vct{x}^T] \succeq
(1-\epsilon) \mtx{I}_{2n} - \frac{1}{2}\begin{bmatrix} \vct{x} \\
  -\bar{\vct{x}} \end{bmatrix} [\vct{x}^*, -\vct{x}^T]
\]
since $\begin{bmatrix} \vct{x} \\ \bar{\vct{x}} \end{bmatrix}
[\vct{x}^*, \vct{x}^T] \succeq \mtx{0}$.  With
$\epsilon=2\delta-\delta^2$, this is the desired conclusion
\eqref{eq:toshow2}.

\end{proof}

\subsection{Dual certificate construction via the golfing scheme}

%

We now construct the approximate dual certificate $\mtx{Z}$ obeying
the conditions of Lemma \ref{lem:ineact}. For this purpose we use the
golfing scheme first presented in the work of Gross
\cite{gross2011recovering}. Modifications of this technique have
subsequently been used in many other papers
e.g.~\cite{candes2011robust, candes2012phaselift, li2012sparse}. The
special form used here is most closely related to the construction in
\cite{li2012sparse}. The mathematical validity of our construction
crucially relies on the lemma below, whose proof is the object of the
separate Section \ref{fixedmat}.
\begin{lemma}
\label{fixedmatlem}
Assume that $L\ge c\log^3n$ for a sufficiently large constant $c$.
Then for any fixed $\mtx{X} \in T$, there exists $\mtx{Y}$ of the form
$\mtx{Y} = \mathcal{A}^*(\vct{\lambda})$ with $\vct{\lambda}$ real
valued such that
\[
\|\mtx{Y}-\mtx{X}\|\leq \frac{\sqrt{2}}{20}\|\mtx{X}\|_F
\]
holds with probability at least $1-1/n^2$. This inequality has the
immediate consequences
\[
\fronorm{\mtx{Y}_T-\mtx{X}}\le\frac{1}{5}\fronorm{\mtx{X}}, \quad
\|\mtx{Y}_{T^\perp}\|\le \frac{\sqrt{2}}{20}\fronorm{\mtx{X}}.
\]
\end{lemma}

To build our approximate dual certificate $\mtx{Z}$, we partition the
modulations or CDPs into $B+1$ different groups so that, from now on,
$\mathcal{A}_0$ corresponds to those measurements from the first $L_0$
modulations, $\mathcal{A}_1$ to those from the next $L_1$ ones, and so
on. Clearly, $L_0+L_1+\ldots+L_B=L$.  The random mappings
$\{\mathcal{A}_b\}_{b=0}^B$ correspond to independent modulations and
are thus independent. Our golfing scheme starts with
$\mtx{X}^{(0)}=\frac{2}{nL_0}\mathcal{P}_T(\cA_0^*(\vct{1}))$
($\vct{1}$ is the all-one vector) and for $b = 1, \ldots, B$,
inductively defines
\begin{itemize}
\item $\mtx{Y}^{(b)}\in\text{Range}(\mathcal{A}_b^*)$ obeying $\|\mtx{Y}^{(b)}-\mtx{X}^{(b-1)}\|\le\frac{\sqrt{2}}{20}\fronorm{\mtx{X}^{(b-1)}}$, 
\item and
  $\mtx{X}^{(b)}=\mtx{X}^{(b-1)}-\mathcal{P}_T(\mtx{Y}^{(b)})$.
\end{itemize}
In the end, we set 
\[
\mtx{Z}=\mtx{Y}-\frac{2}{nL_0}\cA_0^*(\vct{1}), \qquad
\mtx{Y}=\sum_{t=1}^B\mtx{Y}^{(b)}.
\]
Note that Lemma \ref{fixedmatlem} asserts that $\mtx{Y}^{(b)}$ exists
with high probability, and that for each $b$ both
\begin{align}
\label{pariter}
\fronorm{\mtx{X}^{(b)}}\le \frac{1}{5}\fronorm{\mtx{X}^{(b-1)}} \quad
\text{and} \quad \|\mtx{Y}_{T^\perp}^{(b)}\|\le
\frac{\sqrt{2}}{20}\fronorm{\mtx{X}^{(b-1)}}
\end{align}
hold on an event of probability at least $1-1/n^2$.

We now show that our construction $\mtx{Z}$ satisfies the required
assumptions from Lemma \ref{lem:ineact}. First, $\mtx{Z}$ is
self-adjoint and of the form $\mathcal{A}^*(\vct{\lambda})$ with
$\vct{\lambda}\in\R^{nL}$.
Second,
\begin{align*}
  \mtx{Z}_T=\mtx{Y}_T-\frac{2}{nL_0}\mathcal{P}_T(\cA_0^*(\vct{1})) 
  = \sum_{b=1}^B
  \mathcal{P}_T(\mtx{Y}^{(b)})-\mtx{X}^{(0)}=\sum_{b=1}^B
  (\mtx{X}^{(b-1)}-\mtx{X}^{(b)})-\mtx{X}^{(0)}=-\mtx{X}^{(B)}.
\end{align*}
Then \eqref{pariter} implies that with probability at least $1-B/n^2$
\begin{align}
\label{parerr}
\fronorm{\mtx{Z}_T}\le \fronorm{\mtx{X}^{(B)}} \le
\frac{1}{5^B}\fronorm{\mtx{X}^{(0)}}.
\end{align}
Also, \eqref{pariter} gives
\begin{align}
\label{perperr}
\|\mtx{Y}_{T^\perp}\|\leq \sum_{b=1}^B\|\mtx{Y}^{(b)}_{T^\perp}\|\leq \frac{\sqrt{2}}{20}\sum_{b=1}^B\|\mtx{X}^{(t-1)}\|_F\le\frac{\sqrt{2}}{20}\sum_{b=1}^B\frac{1}{5^b}\|\mtx{X}^{(0)}\|_F<\frac{\sqrt{2}}{16}\|\mtx{X}^{(0)}\|_F
\end{align}
with probability at least $1-B/n^2$.  If $L_0\ge c\log n$ for a
sufficiently large constant $c>0$, Lemma \ref{identityinrange} states
that
\begin{align*}
\bigg\|\frac{2}{nL_0}\cA_0^*(\vct{1})-2\mtx{I}\bigg\|\le \frac{1}{4}
\end{align*}
with probability at least $1-1/n^2$.  Using the fact that for any  matrix
$\W$, we have $\|\W_T\| \le 2 \|\W\|$ and $\|\W_{T^\perp}\| \le
\|\W\|$ we conclude that 
\begin{align}
  \|\mtx{X}^{(0)}-2\mtx{I}_T\|\le 1/2, \qquad
  \|\mtx{Y}_{T^\perp}-\mtx{Z}_{T^\perp}-2\mtx{I}_{T^\perp}\|\le
  {1}/{4}.
  \label{interineq1}
\end{align}
Since $\mtx{X}^{(0)}$ has rank at most 2,
\begin{align*}
\fronorm{\mtx{X}^{(0)}}\le\sqrt{2}\big\|\mtx{X}^{(0)}\big\|\le\sqrt{2}\|\mtx{X}^{(0)}-2\mtx{I}_T\|+2\sqrt{2}\big\|\mtx{I}_T\big\|
\end{align*}
Finally, with \eqref{interineq1} and $\|\mtx{I}_T\big\| \le 1$, we
conclude that 
\begin{align}
\label{boundfroX0}
\fronorm{\mtx{X}^{(0)}} < 4.
\end{align}
Plugging this into \eqref{parerr} we arrive at
\begin{align}
\label{parerrZ}
\fronorm{\mtx{Z}_T}\le \frac{4}{5^B}.
\end{align}
Also, \eqref{perperr}, \eqref{interineq1} and \eqref{boundfroX0} give
\begin{align}
\label{perperrZ}
\|\mtx{Z}_{T^\perp}+2\mtx{I}_{T^\perp}\|\le
\|\mtx{Y}_{T^\perp}\|+\|\mtx{Y}_{T^\perp}-\mtx{Z}_{T^\perp}-2\mtx{I}_{T^\perp}\|\le
\frac{\sqrt{2}}{4}+\frac{1}{4}<1\quad\Rightarrow\quad\mtx{Z}_{T^\perp}\preceq
-\mtx{I}_{T^\perp}.
\end{align}
Therefore, the assumptions in Lemma \ref{lem:ineact} hold with
probability at least $1-1/2n$ by applying the union bound and using
with the proviso that $B\ge c_1\log n$ and $L_b\ge c_2\log^3 n$ for
sufficiently large constants $c_1$ and $c_2$ (this is why we require
$L\ge c\log^4 n$ for a sufficiently large constant).

\subsection{Proof of Lemma \ref{fixedmatlem}}
\label{fixedmat}

The immediate consequences hold for the following reasons. First,
since any matrix in $T$ has rank at most 2,
\[
\|\mtx{Y}_T-\mtx{X}\|_F\leq \sqrt{2}\|\mtx{Y}_T-\mtx{X}\| \leq
2\sqrt{2}\|\mtx{Y}-\mtx{X}\|\leq \frac{1}{5}\|\mtx{X}\|_F, 
\]
where the second inequality follows from $\|\mtx{M}_T\| \leq 2 \|\mtx{M}\|$
for any $\mtx{M}$.  Second, since $\|\mtx{M}_{T^\perp}\|\leq
\|\mtx{M}\|$,
\[
\|\mtx{Y}_{T^\perp}\|=\|\mtx{Y}_{T^\perp}-\mtx{X}_{T^\perp}\|\leq
\|\mtx{Y}-\mtx{X}\|\leq \frac{\sqrt{2}}{20}\|\mtx{X}\|_F.
\]
It thus suffices to prove the first property. To this end consider the
eigenvalue decomposition of
$\mtx{X}=\lambda_1\vct{u}_1\vct{u}_1^*+\lambda_2\vct{u}_2\vct{u}_2^*$. The
proof follows from Lemma \ref{initialization} below combined with
Lemma \ref{identityinrange}.

\begin{lemma}
\label{initialization}
Assume $L\ge c \log^3 n$ for a sufficiently large constant $c$.  Given
any fixed self-adjoint matrix $\vct{v}\vct{v}^*$, with probability at
least $1-1/(2n^3)$ there exists $\widetilde{\mtx{Y}} \in
\text{Range}(\mathcal{A}^*)$ obeying
\[
\|\widetilde{\mtx{Y}}-(\vct{v}\vct{v}^*+\twonorm{\vct{v}}^2\mtx{I})\|\leq
\epsilon \twonorm{\vct{v}}^2.
\]
\end{lemma}
\begin{proof}
  Without loss of generality, assume $\twonorm{\vct{v}}=1$ and set
  $\widetilde{\mtx{Y}} = \avg{\mtx{Y}}$, 
\[
\avg{\mtx{Y}} = \frac{1}{L}\sum_{l=1}^L \mtx{Y}_\ell, \qquad
\mtx{Y}_\ell=\frac{1}{n}\sum_{k=1}^n
\left|\vct{f}_k^*\mtx{D}_\ell^*\vct{v}\right|^2\mathbb{1}_{\{\left|\vct{f}_k^*\mtx{D}_\ell^*\vct{v}\right|\leq
  T_n\}}\mtx{D}_\ell\vct{f}_k\vct{f}_k^*\mtx{D}_\ell^*,
\]
which is of the form $\cA^*(\vct{\lambda})$.  The $\mtx{Y}_\ell$'s are
i.i.d.~copies of $\mtx{Y}$,
\begin{align*}
  \mtx{Y} & =\frac{1}{n}\sum_{k=1}^n
  \left|\vct{f}_k^*\mtx{D}^*\vct{v}\right|^2\mtx{D}\vct{f}_k\vct{f}_k^*\mtx{D}^*-\frac{1}{n}\sum_{k=1}^n
  \left|\vct{f}_k^*\mtx{D}^*\vct{v}\right|^2\mathbb{1}_{\{\left|\vct{f}_k^*\mtx{D}^*\vct{v}\right|>T_n\}}\mtx{D}\vct{f}_k\vct{f}_k^*\mtx{D}^*.
\end{align*}
Notice that the random positive semi-definite matrix $\mtx{Y}$ obeys
\[
\mtx{Y}\preceq \frac{1}{n}\sum_{k=1}^n
T_n^2\mtx{D}\vct{f}_k\vct{f}_k^*\mtx{D}^*= T_n^2 \mtx{D}\mtx{D}^*.
\]
By Lemma \ref{teo:expectation1},
\[
\E\left(\frac{1}{n}\sum_{k=1}^n \left|\vct{f}_k^*\mtx{D}^*\vct{v}\right|^2\mtx{D}\vct{f}_k\vct{f}_k^*\mtx{D}^*\right)=\vct{v}\vct{v}^*+\mtx{I}.
\]
Using Jensen's inequality, we have as in the proof of Lemma \ref{lem:PL}
\begin{align}
\label{intermediate}
\left\|\E\left(\frac{1}{n}\sum_{k=1}^n
    \left|\vct{f}_k^*\mtx{D}^*\vct{v}\right|^2\mathbb{1}_{\{\left|\vct{f}_k^*\mtx{D}^*\vct{v}\right|>T_n\}}\mtx{D}\vct{f}_k\vct{f}_k^*\mtx{D}^*\right)\right\|
&\leq \frac{1}{n}\sum_{k=1}^n
\E(\mathbb{1}_{\{\left|\vct{f}_k^*\mtx{D}^*\vct{v}\right|>T_n\}})n^2M^4.
\end{align}
Put $T_n=\sqrt{2\beta\log n}$. Hoeffding's inequality gives 
\[
\E(\mathbb{1}_{\{\left|\vct{f}_k^*\mtx{D}^*\vct{v}\right|>T_n\}})\leq
2n^{-\beta}.
\]
Plugging this into \eqref{intermediate} we arrive at
\[
\left\|\E\left(\frac{1}{n}\sum_{k=1}^n \left|\vct{f}_k^*\mtx{D}^*\vct{v}\right|^2\mathbb{1}_{\{\left|\vct{f}_k^*\mtx{D}^*\vct{v}\right|>T_n\}}\mtx{D}\vct{f}_k\vct{f}_k^*\mtx{D}^*\right)\right\|\leq \frac{2M^4}{n^{\beta-2}}.
\]
For sufficiently large $\beta$, this implies
\begin{align}
\label{meandev}
\left\|\E(\mtx{Y})-(\vct{v}\vct{v}^*+\mtx{I})\right\|\leq \frac{2\|\mtx{D}\|^4}{n^{\beta-2}}\le\frac{\epsilon}{2}.
\end{align}
By using Hoeffding inequality in a similar fashion as in the proof of
Lemma \ref{lem:PL}, we obtain (we omit the details)
\[
\left\|\avg{\mtx{Y}}-\E(\mtx{Y})\right\|\leq \frac{\epsilon}{2}.
\] 
Combining the latter with \eqref{meandev}, we conclude
\[
\left\|\avg{\mtx{Y}}-(\vct{v}\vct{v}^*+\mtx{I})\right\|\leq \epsilon.
\]
\end{proof}

\section{Discussion}

In this paper, we proved that a signal could be recovered by convex
programming techniques from a few diffraction patterns corresponding
to generic modulations obeying an admissibility condition. We expect
that our results, methods and proofs extend to more general random
modulations although we have not pursued such extensions in this
paper. Further, we proved that on the order of $(\log n)^4$ CDPs
suffice for perfect recovery and we expect that further refinements
would allow to reduce this number, perhaps all the way down to a
figure independent of the number $n$ of unknowns. Such refinements
appear quite involved to us and since our intention is to provide a
reasonably short and conceptually simple argument, we leave such
refinements to future research.

\section{Appendix}

Set $\omega=e^{\frac{2\pi i}{n}}$ to be the $n$th root of unity so that 
\begin{align*}
  \vct{f}_k^*=\begin{bmatrix}\omega^{-0(k-1)}, \omega^{-1(k-1)}, \ldots,\omega^{-(n-1)(k-1)}\end{bmatrix},\quad\vct{f}_k=\begin{bmatrix}\omega^{0(k-1)}\\\omega^{1(k-1)}\\\vdots\\\omega^{(n-1)(k-1)}\end{bmatrix}. 
\end{align*}
For two integers $a$ and $b$ we use $a\overset{n}{\equiv}b$ to denote
congruence of $a$ and $b$ modulo $n$ ($n$ divides $a-b$).

\subsection{Proof of Lemma \ref{teo:expectation1}}
Put 
\begin{align*}
\mtx{Y}:=\frac{1}{n}\sum_{k=1}^n\abs{\vct{f}_k^*\mtx{D}^*\vct{x}}^2\mtx{D}\vct{f}_k\vct{f}_k^*\mtx{D}^*.
\end{align*}
By definition,
\begin{eqnarray*}
\abs{\vct{f}_k^*\mtx{D}^*\vct{x}}^2&=&\big(\sum_{a=1}^{n}\bar{d}_{a}x_{a}\omega^{-(a-1)(k-1)}\big)\big(\sum_{b=1}^{n}d_{b}\bar{x}_{b}\omega^{(b-1)(k-1)}\big)\\
&=&\sum_{a=1}^n\sum_{b=1}^n\omega^{(b-a)(k-1)}\bar{d}_ad_bx_a\bar{x}_b
\end{eqnarray*}
Further,
\begin{eqnarray*}
\mtx{Y}_{pq}&=&\frac{1}{n}\sum_{k=1}^n\sum_{a=1}^n\sum_{b=1}^n\omega^{(b-a+p-q)(k-1)}\bar{d}_ad_bd_p\bar{d}_qx_a\bar{x}_b\\
&=&\sum_{a=1}^n\sum_{b=1}^n\bar{d}_ad_bd_p\bar{d}_qx_a\bar{x}_b\Bigl(\frac{1}{n}\sum_{k=1}^n\omega^{(b-a+p-q)(k-1)}\Bigr)\\
&=&\sum_{a=1}^n\sum_{b=1}^n\bar{d}_ad_bd_p\bar{d}_qx_a\bar{x}_b\mathbb{1}_{\{a+q\, \overset{n}{\equiv}\, b+p\}}.
\end{eqnarray*}
Therefore,
\begin{equation*}
\mathbb{E}[\mtx{Y}_{pq}]=\sum_{a=1}^n\sum_{b=1}^n\E[\bar{d}_ad_bd_p\bar{d}_q]x_a\bar{x}_b\mathbb{1}_{\{a+q\, \overset{n}{\equiv}\, b+p\}}.
\end{equation*}
\begin{itemize}
\item \underline{Diagonal terms ($p=q$):} Here, $\E[\bar{d}_ad_b
  |d_p|^2] = 0$ unless $a = b$. This gives
\begin{align*}
\E[\mtx{Y}_{pp}]&=\sum_{a=1}^n\E[\abs{d_a}^2\abs{d_p}^2]\abs{x_a}^2\\
&=\E[\abs{d_p}^4]\abs{x_p}^2+\E[\abs{d_p}^2(\sum_{a\neq p}^n\abs{d_a}^2\abs{x_a}^2)]\\
&=\abs{x_p}^2+\twonorm{\vct{x}}^2.
\end{align*}
\item \underline{Off-diagonal terms ($p\neq q$):} Here
  $\E[\bar{d}_ad_bd_p\bar{d}_q] = 0$ unless $(a=p , b=q)$ so that
\[
\mathbb{E}[\mtx{Y}_{pq}] = {(\mathbb{E}[\abs{d}^2])}^2x_p\bar{x}_q = 
x_p\bar{x}_q.
\]
\end{itemize}
This concludes the proof.

\subsection{Proof of Lemma \ref{teo:expectation2}}
Put
\begin{align*}
\mtx{R}=\frac{1}{n}\sum_{k=1}^n{\big(\vct{f}_k^*\mtx{D}^*\vct{x}\big)}^2\mtx{D}\vct{f}_k\vct{f}_k^T\mtx{D}.
\end{align*}
By definition, 
\begin{equation*}
{(\vct{f}_k^*\mtx{D}^*\vct{x})}^2 
= \sum_{a=1}^n\sum_{b=1}^n\omega^{-(a+b-2)(k-1)}\bar{d}_a\bar{d}_bx_ax_b
\end{equation*}
and
\begin{eqnarray*}
\mtx{R}_{pq}&=&\frac{1}{n}\sum_{k=1}^n\sum_{a=1}^n\sum_{b=1}^n\omega^{(p+q-a-b)(k-1)}\bar{d}_a\bar{d}_bd_pd_qx_ax_b\\
&=&\sum_{a=1}^n\sum_{b=1}^n\bar{d}_a\bar{d}_bd_pd_qx_ax_b\Bigl(\frac{1}{n}\sum_{k=1}^n\omega^{(p+q-a-b)(k-1)}\Bigr)\\
&=&\sum_{a=1}^n\sum_{b=1}^n\bar{d}_a\bar{d}_bd_pd_qx_ax_b\mathbb{1}_{\{p+q\,\overset{n}{\equiv}\,a+b\}}.
\end{eqnarray*}
Therefore,
\begin{equation*}
  \mathbb{E}[\mtx{R}_{pq}]=\sum_{a=1}^n\sum_{b=1}^n\E[\bar{d}_a\bar{d}_bd_pd_q]x_a{x}_b\mathbb{1}_{\{p+q\, \overset{n}{\equiv}\, a+b\}}.
\end{equation*}
\begin{itemize}
\item \underline{Diagonal terms ($p=q$):} Here, $\E[\bar{d}_ad_b
  |d_p|^2] = 0$ unless $a = b = p$. This gives 
\[
\mathbb{E}[\mtx{R}_{pp}] =\mathbb{E}[\abs{d}^4]x_p^2 = 2 x_p^2.
\]
\item \underline{Off-diagonal terms ($p\neq q$):} Here, $\E[\bar{d}_a
  \bar{d_b} d_p d_q] = 0$ unless $(a=p , b=q)$ or $(a=q ,
  b=p)$. This gives 
\[
\mathbb{E}[\mtx{R}_{pq}] = 2\mathbb{E}[\abs{d_p}^2\abs{d_q}^2]x_px_q =
2 x_p x_q.
\]
\end{itemize}
This concludes the proof.

\subsection{Proof of Lemma \ref{identityinrange}}

Note that 
\begin{align*}
  \mtx{Z}:=\frac{1}{nL}\mathcal{A}^*(\vct{1})=\frac{1}{nL}\sum_{\ell=1}^L\sum_{k=1}^n
  \mtx{D}_\ell\vct{f}_k\vct{f}_k^*\mtx{D}_\ell^*=\frac{1}{L}\sum_{\ell=1}^L
  \mtx{D}_\ell \mtx{D}_\ell^*.
\end{align*}
Therefore, $\mtx{Z}$ is a diagonal matrix with i.i.d.~diagonal entries
distributed as $\frac{1}{L}\sum_{\ell=1}^L X_{\ell}$ , where the
$X_\ell$ are i.i.d.~random variables with $\E[X_\ell]=\E[|d|^2]=1$ and
$|X_\ell|=|d|^2\le M^2$. The statement in the lemma then follows from
Hoeffding's inequality
\begin{align*}
\mathbb{P}\bigg\{\abs{\frac{1}{L}\sum_{\ell=1}^L X_{\ell}-1}\ge t \bigg\}\le 2 e^{-\frac{2L}{M^2}t^2}
\end{align*}
combined with the union bound.

\subsection{ Proof of Lemma \ref{upperRIP1}}
The proof is straightforward and parallels calculations in
\cite{candes2012phaselift}. Fix a unit-normed vector $\vct{v}$, then
\begin{align*}
  \onenorm{\mathcal{A}(\vct{v}\vct{v}^*)}=\sum_{\ell=1}^L\sum_{k=1}^n\abs{\vct{f}_k^*\mtx{D}_\ell^*\vct{v}}^2
  = n \sum_{\ell=1}^L\twonorm{\mtx{D}_\ell^*\vct{v}}^2\le
  nM^2\sum_{\ell=1}^L\twonorm{\vct{v}}^2=nLM^2.
\end{align*}
Consider now the eigenvalue decomposition
$\mtx{X}=\sum_{j=1}^n\lambda_j\vct{v}_j\vct{v}_j^*$ where $\lambda_j$
is nonnegative since $\mtx{X} \succeq \mtx{0}$. Then
\begin{align*}
  \onenorm{\mathcal{A}(\mtx{X})}=
  \sum_{j=1}^n\lambda_j\onenorm{\cA(\vct{v}_j\vct{v}_j^*)}\le nLM^2
  \sum_j {\lambda_j} = nLM^2 \tr(\mtx{X}).
\end{align*}

{\small 
\subsection*{Acknowledgements}
E~C. is partially supported by AFOSR under grant FA9550-09-1-0643, by
ONR under grant N00014-09-1-0258 and by a gift from the Broadcom
Foundation.  M.~S.~is supported by a a Benchmark Stanford Graduate
Fellowship. X.~L.~is supported by the Wharton Dean's Fund for
Post-Doctoral Research and by funding from the National Institutes of
Health.  We would like to thank V.~Voroninski for helpful discussions,
especially for bringing to our attention the difficulty of
establishing RIP-1 results in the masked Fourier model. M.~S.~thanks
David Brady and Adam Backer for fruitful discussions about the
implementation of structured illuminations in X-ray crystallography
and microscopy applications.

\bibliography{PhaseLift}
\bibliographystyle{plain}
}
\end{document}